\documentclass[aps,prl,onecolumn,floatfix,superscriptaddress]{revtex4}

\usepackage{bm,bbm,amsmath, amssymb,amsbsy, amsthm,graphicx,subfigure,color,txfonts,multirow}
\usepackage[framemethod=tikz]{mdframed}

 \usepackage[colorlinks]{hyperref}
\hypersetup{
	bookmarksnumbered,
	pdfstartview={FitH},
	citecolor={blue},
	linkcolor={red},
	urlcolor={black},
	pdfpagemode={UseOutlines}
}
\usepackage{cleveref}

\newtheorem*{Theorem}{Theorem}

\newtheorem{Definition}{Definition}
\newtheorem{Proposition}{Proposition}
\newtheorem{Example}{Example}
\newtheorem{Remark}{Remark}
 \newtheorem*{Corollary}{Corollary}

\newcommand{\tr}{{\rm Tr\hskip -0.2em}~}

\newcommand{\mean}{{\rm E}}
\newcommand{\Cov}{{\rm Cov}}
\newcommand{\Var}{{\rm Var}}

\DeclareMathOperator{\frechetdiff}{\mathit d}
\newcommand{\fd}[1]{\hskip-0.2em\frechetdiff\hskip -0.23em{#1}}

\begin{document}
\title{A unified approach to Local Quantum Uncertainty and Interferometric Power by Metric Adjusted Skew Information}

\author{Paolo Gibilisco}
\affiliation{$\hbox{Department of Economics and Finance, University of Rome “Tor Vergata",
Via Columbia 2, Rome 00133, Italy.}$}
\email{$\hbox{paolo.gibilisco@uniroma2.it}$}

\author{Davide Girolami}
\email{davegirolami@gmail.com}
\affiliation{$\hbox{Politecnico di Torino, Corso Duca degli Abruzzi 24, Torino 129, Italy}$}
 
\author{Frank Hansen}
\affiliation{$\hbox{Department of Mathematical Sciences, University of Copenaghen, Universitetsparken 5 DK-2100 Copenhagen, Denmark}$}
\email{$\hbox{frank.hansen@econ.ku.dk}$}

\begin{abstract}{Local quantum uncertainty and interferometric power have been introduced by Girolami {\it et al.} in \cite{GTA:2013,GSGTFSSOA:2014} as geometric quantifiers of quantum correlations. The aim of the present paper is to discuss their properties in a unified manner by means of the the metric adjusted skew information defined by Hansen in \cite{Hansen:2006b}.
}
\end{abstract}

\date{\today}

 
  \maketitle
 
\section{Introduction}

One of the key traits of many-body quantum systems is that the full knowledge of their global configurations does not imply full knowledge of their constituents. The impossibility to reconstruct the local wave functions $|\psi_{1}\rangle,\,|\psi_2\rangle$ (pure states) of two interacting quantum particles from the  wave function of the whole system, $|\psi_{12}\rangle\neq|\psi_{1}\rangle\otimes |\psi_{2}\rangle$, is due to the existence of entanglement \cite{ent}. Investigating open quantum systems, whose (mixed) states are described by density matrices $\rho_{12}=\sum_ip_i|\psi_{i}\rangle_{12}\langle\psi_{i}|$, revealed that the boundary between the classical and quantum worlds is more blurred than we thought. There exists a genuinely quantum kind of  correlation, quantum discord,  which manifests even in absence of entanglement, i.e. in separable density matrices $\rho_{12}=\sum_i p_i \rho_{1,i}\otimes \rho_{2,i}$ \cite{oz,hv}. The discovery triggered theoretical and experimental studies to understand the phy\-sical meaning of quantum discord, and the potential use of  it as a resource for quantum technologies \cite{rev}. Relying on the known interplay between geometrical and physical properties of mixed states \cite{Uhlmann:1992,book}, a stream of works employed information geometry techniques to construct  quantifiers of quantum discord \cite{AyGibiliscoMatus: 2018,BogaertGirolami:2017,CFTA:2018,GII:2009,FPA:2017, GibiliscoIsola:2011}.  In particular, two of the most popular ones are the Local Quantum Uncertainty (LQU) and the Interferometric Power (IP) \cite{GTA:2013,GSGTFSSOA:2014}. A merit of these two measures is that they admit an analytical form for $N$ qubit states across the $1 \, vs\, N-1$ qubit partition. Also, they have a clear-cut physical interpretation. The lack of certainty about quantum measurement outcomes is due to the fact that density matrices are changed by quantum operations. The LQU evaluates the minimum uncertainty about the outcome of a local quantum measurement, when performed on a bipartite system.   It is proven that two-particle density matrices display quantum discord if and only if they are  not ``classical-quantum'' states. That is, they are not  (mixture of) eigenvalues of local observables, $\rho_{12}\neq \sum_ip_i|i\rangle_1\langle i|\otimes \rho_{2,i}$,  or $\rho_{12}\neq \sum_ip_i\rho_{1,i} \otimes |i\rangle_2\langle i|$, in which $\{|i\rangle\}$ is an orthonormal basis. Indeed, this is   the only case in which one can identify a local measurement that does not change a bipartite quantum state, whose spectral decomposition reads $A_1=\sum_i\lambda_i|i\rangle_1\langle i|$, or $A_2=\sum_i\lambda_i|i\rangle_2\langle i|$.  
The LQU was built as the minimum of the Wigner-Yanase skew information,  a well-known information geometry measure \cite{WignerYanase:1963},  between a density matrix and a finite-dimensional observable (Hermitian operator). It quantifies how much a density matrix $\rho_{12}$ is different from being a zero-discord state. The IP was concocted by following a similar line of thinking.  Quantum discord implies a non-classical sensitivity to local perturbations.  This feature of quantum particles, while apparently a limitation, translates into an advantage in the context of quantum metrology \cite{metro}. It was theoretically proven and experimentally demonstrated that quantum systems sharing quantum discord are more sensitive probes for interferometric phase estimation. The figure of merit of such measurement protocols is the quantum Fisher information of the state under scrutiny with respect to a local Hamiltonian (in Information Geometry the QFI is known as the SLD or Bures-Uhlmann metric).
The latter generates a unitary evolution  that imprints information about a physical parameter on the quantum probe. The IP is the minimum quantum Fisher information over all the possible local Hamiltonians, being zero if and only if the probe states are classically correlated. \\

\noindent Here, we polish and extend the mathematical formalization of information-geome\-tric quantum correlation measures. We build a class of parent quantities of the LQU (and consequently of the IP) in terms of the the metric adjusted skew informations \cite{Hansen:2006b}. In Sections 2,3, we review definition and main properties of operator means. In Sections 4-6, we discuss information-geometric quantities that capture complementarity between quantum states and observables. In particular, we focus on the quantum $f$-covariances and the quantum Fisher information. They quantify the inherent uncertainty about quantum measurement outcomes. After having recalled the definition of metric adjusted skew information (Section 7), we build a new  quantum discord measure, the metric adjusted local quantum uncertainty ($f$-LQU), in Section 8. Finally we are able to show that LQU and IP are just two particular members of this family allowing a unified treatment of their fundamental properties.

\section{Means for positive numbers}\label{commutativemeans}

We use the notation ${\mathbb R}_+= (0,+\infty)$. 

\begin{Definition} \label{numbermean}

A bivariate {\sl mean} \cite{PetzTemesi:2005}  is a function $m\colon{\mathbb R}_+ \times {\mathbb R}_+ \to{\mathbb R}_+$  such that

\begin{enumerate}

\item $m(x,x)=x.$ 

\item $m(x,y)=m(y,x).$ 

\item $x <y  $ $\,\Rightarrow\,$ $ x<m(x,y)<y.$ 

\item $x<x' $ and $ y<y' $ $\,\Rightarrow\,$ $ m(x,y)<m(x',y'). $ 

\item $m$ is continuous.

\item $ m $ is positively homogeneous; that is $m(tx,ty)=t \cdot m(x,y)$ for $ t>0. $

\end{enumerate}

\end{Definition}

We use the notation $ \mathcal M_{num} $ for the set of bivariate means described above.

\begin{Definition}
Let $ \mathcal F_{num} $ denote the class of functions $f\colon\mathbb R_+ \to\mathbb R_+$ such that 

\begin{enumerate}

\item $f$ is continuous.

\item $f$ is monotone increasing.

\item $f(1)=1.$

\item $tf(t^{-1})=f(t)$.

\end{enumerate}
\end{Definition}

The following result is straightforward.

\begin{Proposition}
There is a bijection $ f\mapsto m_f $ betwen  ${\mathcal F}_{nu}$  and ${\mathcal M}_{nu}$ given by 
\[
m_f(x,y)=yf(y^{-1}x)\qquad\text{and in reverse}\qquad f(t)=m(1,t)
\]
for positive numbers $ x,y $ and $ t. $
\end{Proposition}



\smallskip

In Table 1 we have some examples of means.

\begin{table}[ht]
\caption{}\label{eqtable}
\renewcommand\arraystretch{2.5}
\renewcommand{\tabcolsep}{8pt}
\noindent\[
\begin{array}{|c|c|c|}
\hline
{\rm Name} & { f} & {m_f}\\
\hline 
{\rm arithmetic} & \displaystyle\frac{1+x}{2} & \displaystyle\frac{x+y}{2}\\
\hline  
{\rm WYD}, \beta\in(0,1)  & \displaystyle\frac{x^{\beta}+x^{1-\beta}}{2} & \displaystyle\frac{x^{\beta}y^{1-\beta}+x^{1-\beta}y^{\beta}}{2} \\
\hline
{\rm geometric} & \sqrt{x} & \sqrt{xy} \\
\hline
{\rm harmonic} & \displaystyle\frac{2x}{x+1} & \displaystyle\frac{2}{x^{-1}+y^{-1}} \\
\hline
{\rm logarithmic} & \displaystyle\frac{x-1}{\log x} & \displaystyle\frac{x-y}{\log x - \log y} \\
\hline
\end{array}
\]
\end{table}



\section{Means for positive operators in the sense of Kubo-Ando} \label{KuboAndomeans}

The celebrated Kubo-Ando theory of operator means 
\cite{KuboAndo79/80, PetzTemesi:2005,GibiliscoHansenIsola:2009} may be viewed as the operator version of the results of Section \ref{commutativemeans}.


\begin{Definition}
A bivariate {\sl mean} $ m $ for pairs of positive operators is a function
\[
(A,B)\to m(A,B),
\]
defined in and with values in positive definite operators on a Hilbert space, that satisfies, mutatis mutandis, conditions $(1)$ to $(5)$ in Definition~\ref{numbermean}. In addition, the
{\sl transformer inequality}
\[
Cm(A,B)C^* \leq m(CAC^*,CBC^*),
\]
should also hold for positive definite $ A, B $ and arbitrary $ C. $ 
\end{Definition}

Note that the transformer inequality replaces condition $ (6) $ in Definition~\ref{numbermean}.
We denote by $\displaystyle {\mathcal M}_{op}$ the set of matrix means.

\begin{Example}
The  arithmetic, geometric and harmonic operator means are defined, respectively, by setting
\[
\begin{array}{rcl}
A \nabla B&=&\frac{1}{2}(A+B)\\[1.5ex]
A\# B&=&A^{1/2}\bigl(A^{-1/2} B A^{-1/2}\bigr)^{1/2}A^{1/2}\\[2ex]
A{\rm !}B&=&2(A^{-1}+B^{-1})^{-1}. 
\end{array}
\]
\end{Example}

We recall that a function $f\colon(0,\infty)\to \mathbb{R}$ is said to be 
{\it operator monotone (increasing)} if
\[
A\le B\quad\Rightarrow\quad f(A)\le f(B)
\]
for positive definite matrices of arbitrary order.  It then follows that the inequality also holds for positive operators on an arbitrary Hilbert space. An operator monotone function $ f $ is said to be {\it symmetric} if
$f(t)=tf(t^{-1})$ for $ t>0 $ and {\it normalized} if $f(1)=1.$

\begin{Definition}

${\mathcal F}_{op}$ is the class of functions $f: {\mathbb R}_+
\to{\mathbb R}_+$ such that
\begin{enumerate}

\item $f$ is operator monotone increasing,

\item $tf(t^{-1})=f(t)\qquad t>0,$

\item $f(1)=1.$

\end{enumerate}
\end{Definition}

The fundamental result, due to Kubo and Ando, is the following.

\begin{Theorem}
There is a bijection $ f\mapsto m_f $ between ${\mathcal M}_{op}$ and ${\mathcal F}_{op}$ given by
the formula
\[
m_f(A,B)= A^{1/2}f(A^{-1/2} BA^{-1/2})A^{1/2}.
\]
\end{Theorem}

\begin{Remark}
The function in ${\mathcal F}_{op}$ are (operator) concave which makes the operator case quite different from the numerical (commutative) case. For example, there exist convex functions in $ \mathcal F_{num}, $ see  \cite{GH:2017}.
\end{Remark}

If $\rho$ is a density matrix (a quantum state) and $A$ is a self-adjoint matrix (a quantum observable), then the expectation of $ A $ in the state $ \rho $ is defined by setting
\[
 \mean_{\rho}(A)= \tr(\rho A).
\]

\section{The correspondence between  Fisher information and metric adjusted skew information}

We introduce now a technical tool which is useful to establish some fundamental relations between quantum covariance, quantum Fisher information and the metric adjusted skew information.

\begin{Definition}

For $f \in {\mathcal F}_{op}$ we define $f(0)=\lim_{x\to 0} f(x).$ We say that a function $f \in {\mathcal
F}_{op}$ is regular if $f(0) \not= 0, $ and non-regular if $f(0)= 0,$ cf.~\cite{PetzSudar:1996,Hansen:2006b}.

\end{Definition}

\begin{Definition}
A quantum Fisher information is extendable if its radial limit exists and it  is a Riemannian metric on the real projective space generated by the pure states.
\end{Definition}

For the definition of the radial limit see \cite{PetzSudar:1996} where the following fundamental result is proved.

\begin{Theorem}
An operator monotone function $f \in {\mathcal F}_{op}$ is regular, if and only if
$ \langle \cdot, \cdot \rangle_{\rho,f}$ is extendable.
\end{Theorem}

\begin{Remark}
The reader should be aware that there is no negative connotation associated with the qualification ``non-regular".
For example, a very important quantum Fisher information in quantum physics (see \cite{FickSauermann:1990}), namely the Kubo-Mori metric related to the function $f(x)=(x-1)/\log x,$ is non-regular.
\end{Remark}

We introduce the sets of regular and non-regular functions
\[
{\mathcal F}_{op}^{\, r}:=\{f\in {\mathcal F}_{op}\mid f(0) \not= 0 \},  \quad
{\mathcal F}_{op}^{\, n}:=\{f\in {\mathcal F}_{op}\mid f(0) = 0 \}
\]
and notice that trivially ${\mathcal F}_{op}={\mathcal F}_{op}^{\, r}\dot\cup{\mathcal F}_{op}^{\, n}$\,.

\begin{Definition}
We introduce to $f \in {\mathcal F}_{op}^{\, r}$ the transform $ \tilde f $ given by
\[
\tilde{f}(x)=\frac{1}{2}\left[ (x+1)-(x-1)^2 \frac{f(0)}{f(x)}
\right]
\]
for $ x>0. $ We may also write $ {\tilde f}={\mathcal G}(f), $ cf. {\rm \cite{GibiliscoImparatoIsola:2007, GibiliscoHansenIsola:2009}}.
\end{Definition}
The following result is taken from \cite[Theorem 5.1]{GibiliscoHansenIsola:2009}.

\begin{Theorem}\label{correspondence theorem}
The correspondence $ f \to \tilde f $ is a bijection between ${\mathcal F}_{op}^{\, r}$ and ${\mathcal F}_{op}^{\, n}\,. $
\end{Theorem}

In \textsc{Table 2} we have some examples (where $0<\beta<1$).
\begin{table}[ht]
\caption{}\label{eqtable}
\renewcommand\arraystretch{2.5}
\renewcommand{\tabcolsep}{8pt}
\noindent\begin{align*}
\begin{array}{|c|c|}
\hline
 f & \tilde f\\
\hline 
 \displaystyle\frac{1+x}{2} &  \displaystyle\frac{2x}{x+1} \\
 \hline
  \displaystyle \frac{(\sqrt{x}+1)^2}{4}& \sqrt{x} \\
\hline  
 \displaystyle \beta (1-\beta) \frac{(x-1)^2}{(x^{\beta}-1) (x^{1-\beta}-1)} &  \displaystyle \frac{x^{\beta}+x^{1-\beta}}{2} \\
\hline
\end{array}
\end{align*}
\end{table}


\section{Quantum f-Covariance}

The notion of quantum $f$-covariance has been introduced by Petz, see \cite{Petz:2003,GibiliscoHiaiPetz:2009}.
Any Kubo-Ando function $ m_f(x,y)=yf(y^{-1}x) $ for $ x,y>0 $ has a continuous extension to $[0,+\infty) \times [0,+\infty)$ given by
\[
m_f(0,y)=f(0) y, \quad m_f(x,0)=f(0) x, \quad m_f(0,0)=0, \quad x,y>0.
\]
The operator $ m_f(L_{\rho},R_{\rho}) $ is well-defined by the spectral theorem
for any state, see \cite[Proposition 11.1 page 11]{GibiliscoImparatoIsola:2007}.
To self-adjoint $ A $ we set $ A_0=A-(\tr\rho A) I, $ where $ I $ is the identity operator. Note that
\[
\tr \rho A_0=\tr\rho A-(\tr\rho A)\tr \rho=0,
\]
if $ \rho $ is a state.

\begin{Definition} Given a state $ \rho, $ a function $ f\in \mathcal F_{op}  $ and self-adjoint $ A,B $ we define the quantum $ f $-covariance by setting
\[
\Cov_{\rho}^f(A,B) = \tr B_0\, m_f(L_{\rho}, R_{\rho}) A_0
\]
and the corresponding quantum $ f $-variance by $
\Var_{\rho}^f(A)=\Cov_{\rho}^f(A,A). 
$
\end{Definition}

The $ f $-variance is a positive semi-definite sesquilinear form and
\begin{equation}\label {increasingVar}
f \leq g \quad \Rightarrow \quad \Var_{\rho}^f(A) \leq \Var_{\rho}^g(A).
\end{equation}
Note that for the standard covariance we have $\Cov_{\rho}(A,B)=\Cov_{\rho}^{SLD}(A,B),$ where the SLD or Bures-Uhlmann metric is the one associated with the function $(1+x)/2$.

\begin{Proposition} If $\rho$ is a pure state then
$
\Var_{\rho}^f(A)= 2\,m_f(1,0)  \cdot \Var_{\rho}(A),
$ cf. \cite{TothPetz:2013}.
\end{Proposition}

\begin{Corollary}\label{Variancepurestate}

If $\rho$ is a pure state and $f$ is non-regular,  then
$
\Var_{\rho}^f(A)= 0.
$

\begin{proof}
If $f$ is non regular $m_f(1,0)=0$
\end{proof}

\end{Corollary}

\section{Quantum Fisher Information}
The theory of quantum Fisher information is due to Petz and we recall here the basic results.
If ${\mathcal N}$ is a differentiable manifold we denote
by $T_{\rho} \mathcal N$ the tangent space to $\mathcal N$ at the point
$\rho \in {\mathcal N}$.  Recall that there exists a natural
identification
 of $T_{\rho}{\mathcal D}^1_n$ with the space of self-adjoint traceless
 matrices; namely, for any $\rho \in {\mathcal D}^1_n $
$$
T_{\rho}{\mathcal D}^1_n =\{A \in M_n\mid A=A^* \, , \, \hbox{Tr}\, A=0 \}.
$$
A stochastic map is a completely positive and trace preserving operator $T:
M_n \to M_m$. A {\sl monotone metric} is a family of Riemannian metrics $g=\{g^n\}$
 on $\{{\mathcal D}^1_n\}$, $n \in \mathbb{N}$, such that
 $$
 g^m_{T(\rho)}(TX,TX) \leq g^n_{\rho}(X,X)
 $$
 holds for every stochastic map $T:M_n \to M_m$, every faithful state $\rho \in
 {\mathcal D}^1_n,$ and every $X \in T_\rho {\mathcal D}^1_n$.
Usually monotone metrics are normalized in such a way that
$[A,\rho]=0$ implies $g_{\rho} (A,A)={\rm Tr}({\rho}^{-1}A^2)$.
A monotone metric is also called (an example of) {\sl quantum Fisher information} (QFI). This notation
is inspired by Chentsov's uniqueness theorem for commutative monotone metrics \cite{Chentsov:1982}.

Define $L_{\rho}(A)= \rho A$ and $R_{\rho}(A)= A\rho$, and observe
 that $ L_\rho $ and $ R_\rho $ are commuting positive superoperators on $M_n.$ For any $f\in {\mathcal F}_{op}$ one may also define the positive (non-linear) superoperator
$m_f(L_{\rho},R_{\rho})$.
The fundamental theorem of monotone metrics may be stated in the following way:

\begin{Theorem} (See \cite{Petz:1996}).
    There exists a bijective correspondence between monotone metrics (quantum Fisher information(s))
    on ${\mathcal D}^1_n$ and functions $f\in {\mathcal F}_{op}$.  The correspondence is given by
    the formula
    $$
   \langle A,B \rangle_{\rho,f}={\rm Tr}(A\cdot
    m_f(L_{\rho},R_{\rho})^{-1}(B))
    $$
    for positive matrices $ A $ and $ B. $
\end{Theorem}

\section{Metric adjusted skew information}

By using the general form of the quantum Fisher information it is possible to greatly generalize the Wigner-Yanase information measure. To $ f \in \mathcal F_{op} $ the so-called Morosova function $ c_f(x,y) $ is defined by setting
\begin{equation}
c_f(x,y)=\frac{1}{yf(xy^{-1})}=m_f(x,y)^{-1}\qquad x,y>0.
\end{equation}
The corresponding monotone symmetric metric $ K_\rho $ is given by
\begin{equation}
K_\rho^f(A,B)=\tr A^* c_f\bigl(L_\rho,R_\rho\bigr) B,
\end{equation}
where $ L_\rho $ and $ R_\rho $ denote left and right multiplication with $ \rho. $  Note that 
$ K^f_{\rho}(A) $ is increasing in $ c_f $ and thus decreasing in $ f. $
If furthermore $ f $ is regular, the notion of metric adjusted skew information \cite[Definition 1.2]{Hansen:2006b} is defined by setting
\begin{equation}
I^f_{\rho}(A) =I^f(\rho,A)=\frac{f(0)}{2} K^f_\rho\bigl(i[\rho,A^*],  i[\rho,A]\bigr),
\end{equation}
where $ \rho>0. $ We use the second notation, $ I^f(\rho,A), $ when the expression of the state takes up too much space. We also tacitly extended the metric adjusted skew information to arbitrary (non-self-adjoint) operators $ A. $
It is convex  \cite[Theorem 3.7]{Hansen:2006b}  in the state variable $ \rho $ and
\begin{equation}
0\le I^f_\rho(A)\le \Var_\rho(A)
\end{equation}
with equality if $ \rho $ is pure  \cite[Theorem 3.8]{Hansen:2006b}, see also the summery with interpretations in \cite[Theorem 1.2]{CaiHansen:2010}.  Furthermore, the notion of unbounded metric adjusted skew information for non-regular functions in $ \mathcal F_{op} $ is introduced in \cite[Theorem 5.1]{CaiHansen:2010}.  For regular $ f\in\mathcal F_{op} $ the metric adjusted skew information may be written as
\[
I_\rho^f(A)=\tr\rho A^2-\tr A\, m_{\tilde f}(L_\rho,R_\rho) A,
\]
se \cite[equation (7)]{AudenaertCaiHansen:2008}. We thus obtain that the metric adjusted skew information is decreasing in the transform $ \tilde f $ for arbitrary self-adjoint $ A, $ that is
\begin{equation} \label{decreasingMASI}
\tilde f \leq \tilde g \quad \Rightarrow \quad I_{\rho}^f (A) \geq I_{\rho}^g(A)\qquad\text{for}\quad f,g\in {\mathcal F}_{op}^{\, r}\,.
\end{equation}
We may also write
\[
\check{f}=\frac{f(0)}{f(t)}\qquad\text{and}\qquad \check{c}(x,y)=y^{-1}\check{f}(xy^{-1})
\]
and obtain
\[
I_\rho^f(A)=\frac{1}{2} \tr i[\rho,A^*] \check{c}\bigl(L_\rho,R_\rho\bigr) i[\rho,A],
\]
cf. \cite[equation (10)]{AudenaertCaiHansen:2008}. It follows that the metric adjusted skew information is increasing in $ \check{f} $ for arbitrary $ A. $
 It may be derived from \cite[Proposition 6.3, page 11]{GibiliscoImparatoIsola:2007}, that the metric adjusted skew information can be expressed as the difference
\[
I^f_{\rho}(A)=\Var_{\rho}(A) - \Var^{\tilde f}_{\rho}(A)
\]
with extension to the sesquilinear form
\[
I^f_{\rho}(A,B)=\Cov_{\rho}(A,B) - \Cov^{\tilde f}_{\rho}(A,B).
\]

\subsection{Information inequalities}

A function $ f\colon \mathbb{R_+}\to \mathbb{R_+} $ is in $ \mathcal F_{op} $ if and only if it allows a representation of the form
\begin{equation}\label{canonical representation in terms of a weigt-function}
f(t)=\frac{1+t}{2}\exp\int_0^1\frac{(\lambda^2-1)(1-t)^2}{(\lambda+t)(1+\lambda t)(1+\lambda)^2}\,h_f(\lambda)\,d\lambda,
\end{equation}
where the weight function $ h_f\colon [0,1]\to[0,1] $ is measurable. The equivalence class containing $ h_f $ is uniquely determined by $ f, $
cf. \cite[Theorem 2.1]{AudenaertCaiHansen:2008}. This representation gives rise to an order relation in $ \mathcal F_{op}. $

\begin{Definition}
Let $ f,g\in \mathcal F_{op}. $ We say that $ f $ is majorized by $ g $ and write $ f\preceq g, $ if the function
\[
\varphi(t)=\frac{t+1}{2}\,\frac{f(t)}{g(t)}\qquad t>0
\]
is in $ \mathcal F_{op}\,. $ 
\end{Definition}

The partial order relation $ \preceq $ is stronger that the usual order relation $ \le, $ and
it renders $ (\mathcal F_{op}\,, \preceq) $ into a lattice with
\begin{equation}
f_\text{min}(t)=\frac{2t}{t+1}\qquad\text{and}\qquad f_\text{max}(t)=\frac{t+1}{2}\qquad
\end{equation}
as respectively minimal element and maximal element.
Furthermore, 
\begin{equation}
 f\preceq g\quad\text{if and only if}\quad h_f\ge h_g\qquad\text{almost everywhere}, 
\end{equation} 
 cf.  \cite[Theorem 2.4]{AudenaertCaiHansen:2008}. The restriction of $ \preceq $ to the regular part of $ \mathcal F_{op} $ induces a partial order relation $ \preceq $ on the set of metric adjusted skew informations.
 
 \begin{Proposition}\label{restrictions of the lattice structure}
The restriction of the order relation $ \preceq $ renders the regular part of $  \mathcal F_{op} $ into a lattice. In addition, if one of two functions $ f, g\in\mathcal F_{op} $ is non-regular, then the minorant $ f\wedge g  $ is also non-regular.
\end{Proposition}

\begin{proof}
Take $ f\in \mathcal F_{op} $ with representative function $ h_f $ as given in (\ref{canonical representation in terms of a weigt-function}). Then it follows that $ f $  is regular if and only if the integral
\begin{equation}\label{integration condition}
\int_0^1 \frac{h_f(\lambda)}{\lambda}\, \fd{}\lambda<\infty.
\end{equation}
Take now regular functions $ f,g\in\mathcal F_{op\,.} $ We know that $ \bigl(\mathcal F_{op}\,,\preceq\bigr) $ is a lattice \cite[bottom of page 141]{AudenaertCaiHansen:2008}, and that the
representative function in (\ref{canonical representation in terms of a weigt-function}) for the minorant $ f\wedge g $ is given by 
\[
h_{f\wedge g}=\max\{h_f,h_g\} \le h_f+h_g
\]
showing that also $ h_{f\wedge g} $ satisfies the integrability condition (\ref{integration condition}) implying that $ f\wedge g $ is regular. Since
\[
h_{f\vee g}=\min\{h_f,h_g\} \le h_f
\]
it also follows that the majorant is regular.

We now take functions $ f,g\in \mathcal F_{op} $  with representative functions $ h_f $ and $ h_g $ and assume that $ f $ is non-regular. Since 
\[
h_{f\wedge g}=\max\{h_f,h_g\}\qquad\text{and thus}\qquad h_f\le h_{f\wedge g}
\]
we obtain that also the minorant $ f\wedge g $ is non-regular.
\end{proof}

\subsection{The Wigner-Yanase-Dyson skew informations}

The Wigner-Yanase-Dyson skew information (with parameter $ p) $ is defined by setting
\[
I_p(\rho,A)=-\frac{1}{2} \tr [\rho^p,A[[\rho^{1-p},A],\qquad 0<p<1.
\]
It is an example of a metric adjusted skew information and reduces to the Wigner-Yanase skew information for $ p=1/2\,. $ The representing function $ f_p $ of $ I_p(\rho,A) $ is given by
\[
f_p(t)=p(1-p)\cdot\frac{(t-1)^2}{(t^p-1)(t^{1-p}-1)}\qquad 0<p<1,
\]
that is $ I_p(\rho,A) = I^{f_p}_{\rho}(A). $ The weight-functions $ h_p(\lambda) $ in equation (\ref{canonical representation in terms of a weigt-function}) corresponding to the representing functions $ f_p $ are given by
\[
h_p(\lambda)=\frac{1}{\pi}\arctan\frac{(\lambda^p +
\lambda^{1-p})\sin p\pi}{1-\lambda-(\lambda^p - \lambda^{1-p})\cos
p\pi}\qquad 0<\lambda<1.
\]
It is a non-trivial result that the Wigner-Yanase-Dyson skew informations $ I_p(\rho,A) $ are increasing in the parameter $ p $ for $ 0< p\le 1/2 $ and decreasing in $ p $ for $ 1/2\le p <1 $ with respect to the order relation $ \preceq, $ cf. \cite[Theorem 2.8]{AudenaertCaiHansen:2008}. The Wigner-Yanase skew information is thus the maximal element among the Wigner-Yanase-Dyson skew informations with respect to the order relation  $ \preceq. $

\subsection{The monotonous bridge}

The family of metrics with representing functions
\[
f_\alpha(t)=t^\alpha\left(\frac{1+t}{2}\right)^{1-2\alpha}\qquad t>0,
\]
decrease monotonously (with respect to $ \preceq) $ from the largest monotone symmetric metric down to the Bures metric for $ \alpha $ increasing from $ 0 $ to $ 1. $ They correspond the the constant weight functions $ h_\alpha(\lambda)=\alpha $ in  equation (\ref{canonical representation in terms of a weigt-function}).
However, the only regular metric in this bridge is the Bures metric $ (\alpha=1). $ It is however possible to construct a variant bridge by choosing the weight functions
\[
    h_p(\lambda)=\left\{\begin{array}{lrl}
                        0,\quad &\lambda &<1-p\\[1ex]
                        p, &\lambda&\ge 1-p
                        \end{array}\right.\qquad 0\le p\le 1
    \]
    in equation (\ref{canonical representation in terms of a weigt-function})
    instead of the constant weight functions. It is non-trivial that these weight functions provide a monotonously decreasing bridge (with respect to $ \preceq) $ of monotone symmetric metrics between the smallest and the largest (monotone symmetric) metric. The benefit of this variant bridge is that all the constituent metrics are regular except for $ p=1. $

\section{Metric adjusted local quantum uncertainty}

We consider a bipartite system $ \mathcal H=\mathcal H_1\otimes\mathcal H_2 $ of two finite dimensional Hilbert spaces. 

\begin{Definition} Let $ f\in \mathcal F_\text{op} $ be regular and take a vector $ \Lambda\in\mathbf R^d. $ We define the {\em Metric Adjusted Local Quantum Uncertainty} (or $f$-LQU) by setting
\begin{equation}
\mathcal U_1^{\Lambda,f}(\rho)= \inf\{ I^f_\rho(K_1\otimes 1_2)\mid K_1\text{ has spectrum $ \Lambda $}\} ,
\end{equation}\label{definition of f-LQU}
where $ \rho_{12} $ is a bipartite state, and $ K_1 $ is the partial trace of an observable $ K $ on $  \mathcal H. $
\end{Definition}

The minimum in the above definition is thus taken over local observables $ K_1\otimes 1_2\in B(\mathcal H_1\otimes H_2) $ such that $ K_1 $ is unitarily equivalent with the diagonal matrix $ \text{diag}(\Lambda). $

\begin{Remark} The metric adjusted LQU has been studied in the literature for specific choices of $ f. $

\begin{itemize}

\item If $ f(x)=f_{WY}(x)=\Bigl(\frac{1+\sqrt{x}}{2} \Bigr)^2$ then ${\mathcal U}_1^{\Lambda,f} $ coincide with the LQU introduced in  \cite[equation 2]{GTA:2013}.

\item If $ f(x)=f_{SLD}(x)=\frac{1+x}{2} $ then $ \mathcal U_1^{\Lambda,f} $ coincides with the Interferometric Power introduced in \cite{GSGTFSSOA:2014}.

\end{itemize}
\end{Remark}

\begin{Proposition}\label{inequality for LQU}  
For $ f ,g \in \mathcal F_{op}^{\,r} $ with $\tilde g \leq \tilde f $ we have the inequality
$
{\mathcal U}^{\Lambda, f}_1(\rho_{12})
\leq
{\mathcal U}^{\Lambda, g}_1(\rho_{12}). 
$
In particular the LQU is smaller than the IP.
\end{Proposition}

\begin{proof}
Let $ \tilde K_1 $ be the local observable with spectrum $ \Lambda $ minimizing the metric adjusted skew information. Then
\[
{\mathcal U}^{\Lambda, f}_1(\rho_{12}) =  I^f_{\rho_{12}}\bigl(\tilde K_1\otimes 1_2\bigr)\ge  I^g_{\rho_{12}}\bigl(\tilde K_1\otimes 1_2\bigr)\ge {\mathcal U}^{\Lambda, g}_1(\rho_{12}),
\]
where we used the inequality in  (\ref{decreasingMASI}).
\end{proof}

\begin{Corollary} Let $ g_1 $ and $ g_2 $ be regular functions in $ \mathcal F_{op} $ and set
$
f=\tilde g_1\wedge\tilde g_2
$
with respect to the lattice structure in $ \mathcal F_{op}\,. $ Then there is a regular function $ g $ in 
$ \mathcal F_{op} $ such that $ \tilde g=f=\tilde g_1\wedge\tilde g_2 $ and
\[
\max\bigl\{{\mathcal U}^{\Lambda, g_1}_1(\rho_{12}),\, {\mathcal U}^{\Lambda, g_2}_1(\rho_{12})\bigr\}
\le {\mathcal U}^{\Lambda, g}_1(\rho_{12})
\]
for arbitrary $ \rho_{12}\,. $
\end{Corollary}

\begin{proof}
The functions $ \tilde g_1 $ and $ \tilde g_2 $ are non-regular by Theorem~\ref{correspondence theorem}. By Proposition \ref{restrictions of the lattice structure} we thus obtain that also the minorant $ f $ is non-regular. Therefore there exists, by the correspondence in 
Theorem~\ref{correspondence theorem}, a (unique) regular function $ g $  in $ \mathcal F_{op} $ such that $ \tilde g=f. $ The assertion then follows by Proposition~\ref{inequality for LQU}.   
\end{proof}

Following \cite{BogaertGirolami:2017} we prove that the metric adjusted LQU is a measure of non-classical correlations, i.e. it meets the criteria which identify discord-like quantifiers, see \cite{rev}.

\begin{Theorem}
If the state $ \rho $ is classical-quantum in the sense of \cite{Piani},
then the metric adjusted LQU vanishes, that is
$ \mathcal U_1^{\Lambda,f}(\rho)=0. $ 
Conversely, if the coordinates of $ \Lambda $ are mutually different (thus rendering the operator $ K_1 $ non-degenerate) and $\mathcal U_1^\Lambda(\rho) =0, $ then $ \rho $ is classical-quantum. 
\end{Theorem}

\begin{proof}
We note that the metric adjusted skew information $ I_\rho^f(A) $ for a faithful state $ \rho $ is vanishing if and only if $ \rho $ and $ A $ commute. If $ \rho $ is classical-quantum, then
\[
P_1(\rho)=\sum_i (P_{1,i}\otimes 1_2)\rho(P_{1,i}\otimes 1_2)=\rho
\]
for some von Neumann measurement $ P $ given by a resolution $ (P_i) $ of the identity $ 1_1 $ in terms of one-dimensional projections. We may choose $ K_1 $ diagonal with respect to this resolution, so $ K_1\otimes 1_2 $ and $ \rho $ commute and thus $ \mathcal U_1^{\Lambda,f}(\rho) =0. $

If on the other hand the $f$-LQU $\mathcal U_1^{\Lambda,f}(\rho)=0, $ then there exist a local observable $ K_1\otimes 1_2 $ such that $ [\rho,K_1\otimes 1_2]=0. $ By the spectral theorem we write
\[
K_1=\sum_i \lambda_i P_{1,i}=\sum_i\lambda_i|i\rangle_1\langle i| 
\]
and since
\[
\rho(K_1\otimes 1_2)=(K_1\otimes 1_2)\rho
\]
we obtain by multiplying with $ P_{1,i}\otimes 1_2 $ from the left and $ P_{1,j}\otimes 1_2 $ from the right the identity
\[
\lambda_j (P_{1,i}\otimes 1_2)\rho(P_{1,j}\otimes 1_2)
=\lambda_i (P_{1,i}\otimes 1_2)\rho(P_{1,j}\otimes 1_2).
\]
If $ K_1 $ is non-degenerate, it thus follows that
\[
(P_{1,i}\otimes 1_2)\rho(P_{1,j}\otimes 1_2)=0\qquad\text{for}\quad i\ne j. 
\]
By summing over all $ j $ different from $ i $ we obtain
\[
(P_{1,i}\otimes 1_2)\rho((1_1-P_{1,i})\otimes 1_2)=0,
\]
thus 
\[
(P_{1,i}\otimes 1_2)\rho=(P_{1,i}\otimes 1_2)\rho(P_{1,i}\otimes 1_2),
\]
so $ P_{1,i}\otimes 1_2 $ and $ \rho $ commute. It follows that
\[
P_1(\rho)=\sum_i (P_{1,i}\otimes 1_2)\rho(P_{1,i}\otimes 1_2)=\rho,
\]
so $ \rho $ is left invariant under the von Neumann measurement $ P $ given by $ (P_i). $ Therefore, $ \rho $ is classical-quantum.
\end{proof}

Recall that Luo and Zhang \cite{LuoZang:2008}
proved that a state $ \rho $ is classical-quantum if and only if there exists a resolution $(P_i)$ of the identity $ 1_1 $ such that 
\[
\rho=\sum_i p_i P_{1,i}\otimes \rho_i,
\]
where each $ \rho_i $ is a state on $ \mathcal H_2 $ and $ p_i\ge 0,  $ and the sum $ \sum_i p_i=1. $
By  \cite[Lemma 3.1]{CaiHansen:2010}  the inequality
\[
I_\rho^f(K_1\otimes 1_2)\ge I_{\rho_1}^f(K_1)
\]
is valid, where $ \rho_1=\tr_2\,\rho_{12}\,. $ Consequently, we obtain that
\begin{equation}
\mathcal U_1^{\Lambda,f}(\rho)\ge  \inf_{K_1} \, I^f_{\rho_1}(K_1)=\inf_{\sigma_1} \, I^f_{\sigma_1}\bigl(K_1\bigr),
\end{equation}
where the minimum is taken over states $ \sigma_1 $ on $ \mathcal H_1 $ unitarily equivalent with $ \rho_1. $

\begin{Theorem}
The metric adjusted LQU is invariant under local unitary transformations.
\end{Theorem}

\begin{proof}
For the metric adjusted skew information and local unitary transformations we have
\[
\begin{array}{l}
\mathcal U_1^{\Lambda,f}\bigl((U_1  \otimes U_2) \rho_{12}(U_1  \otimes U_2)^\dagger\bigr)\\[1.5ex]
=\min_{K_1}I^f\bigl((U_1  \otimes U_2) \rho_{12}(U_1  \otimes U_2)^\dagger, K_1\otimes 1_2\bigr)\\[1.5ex]
=\min_{K_1}I^f\bigl(\rho_{12},  (U_1  \otimes U_2)^\dagger(K_1 \otimes 1_2)(U_1  \otimes U_2) \bigr)\\[1.5ex]
=\min_{K_1} I^f\bigl(\rho_{12}, (U_1 ^{\dagger}K_1 U_1 \otimes 1_2\bigr)
=\mathcal U^{\Delta,f}_1(\rho_{12}),
\end{array}
\]
where we used the definition in (\ref{definition of f-LQU}).
\end{proof}

\begin{Theorem}
The metric adjusted LQU is contractive under completely positive trace-preserving maps on the non-measured subsystem.
\end{Theorem}

\begin{proof}
Let $ \tilde K_1 $ be the local observable minimizing the metric adjusted skew information. 
A completely positive trace preserving map $ \Phi_2 $ on system 2 is obtained as an amplification followed by a partial trace (Stinespring dilation): $\tr_3(U_{23}\rho_{23}U_{23}^\dagger)= \Phi_2\, \rho_2$. The metric adjusted LQU  is invariant under local unitaries. Also, the metric adjusted skew information is contractive under partial trace. Calling $d_3$ the dimension of the Hilbert space of the ancillary system 3, one has 
\begin{align*}
\mathcal U_1^{\Lambda,f}(\rho_{12}) =&\,  I^f\left(\rho_{12}, \tilde K_1\otimes 1_{2}\right) =I^f\left(\rho_{12} \otimes  \frac1{d_3} 1_3, \tilde K_1\otimes 1_{23}\right)  \\
=&\,\ I^f\left((1_1\otimes U_{23})\left(\rho_{12}\otimes \frac1{d_3} 1_3\right)(1_1\otimes U_{23}^\dagger),\tilde K_1\otimes 1_{23}\right)  \\
\ge&\, I^f\left(\tr_3\left\{(1_1\otimes U_{23})\left(\rho_{12}\otimes \frac1{d_3} 1_3\right)(1_1\otimes U_{23}^\dagger)\right\},\tilde K_1\otimes 1_{2}\right)\\
=&\,I^f\left((1_1\otimes \Phi_2)\rho_{12}, \tilde K_1\otimes 1_{2}\right)\\
\ge&\, \mathcal U_1^{\Lambda,f}\bigl((1_1\otimes \Phi_2)\rho_{12}\bigr),
\end{align*}
as desired.
\end{proof}

\begin{Theorem}
The metric adjusted LQU reduces to an entanglement monotone for pure states.
\end{Theorem}

\begin{proof}
The metric adjusted $ f $-LQU  coincides with the standard variance on pure states, that is
\[
I_\rho^f(A)=\Var_\rho(A)=\tr\rho A^2 - (\tr\rho A)^2
\]
whenever $ \rho $ is pure \cite[Theorem 3.8]{Hansen:2006b}.  But in  \cite{GTA:2013}
it has been proved that the minimum local variance is an entanglement monotone for pure states.
\end{proof}

\section{Conclusion}
In this work, we have built a unifying information-geometric framework to quantify quantum correlations in terms of metric adjusted skew informations. We extended the physically meaningful definition of LQU to a more general class of information measures. Crucially, metric adjusted quantum correlation quantifiers enjoy, by construction, a set of desirable properties which make them robust information measures.\\
An important open question is whether information geometry methods may help characterize many-body quantum correlations. In general, the very  concept of multipartite statistical dependence is not fully grasped in the quantum scenario. In particular, we
do not have axiomatically consistent and operationally meaningful measures of genuine multipartite quantum
discord.  Unfortunately, the LQU and IP cannot be straightforwardly generalized to capture joint properties of more than two quantum particles. A promising starting point could be to translate into the information-geometry language the  entropic multipartite correlation measures developed in \cite{weaving}. We plan to investigate the issue in future studies.\\[1ex]

\vspace{6pt}

\acknowledgments{This research is  supported by a Rita Levi Montalcini Fellowship of
the Italian Ministry of Research and Education (MIUR), grant
number 54$\_$AI20GD01.}

 



\end{document}